\documentclass[letterpaper, 10 pt, conference]{ieeeconf}
\usepackage{cite}

 \usepackage[pdftex]{graphicx}

\usepackage{amsmath}
\usepackage{xcolor}

\usepackage{subcaption}
\usepackage{mathtools}

\newtheorem{theorem}{Theorem}
\newtheorem{assumption}{Assumption}

\begin{document}
%
\title{\LARGE \bf Collaborative Platooning of Automated Vehicles\\ Using Variable Time-Gaps}

\author{Aria HasanzadeZonuzy$^*$,
Sina Arefizadeh$^\dagger$,
Alireza Talebpour$^\dagger$, 
Srinivas Shakkottai$^*$ and
Swaroop Darbha$^\ddagger$\\
$^*$Dept. of Electrical and Computer Engineering, Texas A\&M University\\
$^\dagger$Dept. of Civil Engineering, Texas A\&M University\\
$^\ddagger$Dept. of Mechanical Engineering, Texas A\&M University\\
{\small Email:\{azonuzy, sinaarefizadeh, atalebpour, sshakkot, dswaroop\}@tamu.edu}}


\pdfminorversion=4


\maketitle


\begin{abstract}
Connected automated vehicles (CAVs) could potentially be coordinated to safely attain the maximum traffic flow on roadways under dynamic traffic patterns, such as those engendered by the merger of two strings of vehicles due a lane drop.  Strings of vehicles have to be shaped correctly in terms of the inter-vehicular time-gap and velocity to ensure that such operation is feasible.  However, controllers that can achieve such traffic shaping over the multiple dimensions of target time-gap and velocity over a region of space are unknown.   The objective of this work is to design such a controller, and to show that we can design candidate time-gap and velocity profiles such that it can stabilize the string of vehicles in attaining the target profiles.  Our analysis is based on studying the system in the spacial rather than the time domain, which enables us to study stability as in terms of minimizing errors from the target profile and across vehicles as a function of location.  Finally, we conduct numeral simulations in the context of shaping two platoons for merger, which we use to illustrate how to select time-gap and velocity profiles for maximizing flow and maintaining safety.
\end{abstract}


%


\section{Introduction}
\label{sec:intro}



Traffic shaping in terms of achieving desired time-gaps and velocities over platoons of vehicles is needed to handle variable traffic flows on highways caused by mergers at highway entrances, departures at exits or prevailing road conditions such as a lane drop.  If uncontrolled, such events could lead to shockwave formation and breakdown.  Connected automated vehicles (CAVs) have the potential to maintain maximum traffic flow on roadways under such dynamically changing traffic patterns.  However, even in the CAV setting, significant coordination is needed to ensure that traffic is shaped in a manner that allows the seamless merger of the vehicle platoons in a safe manner.



\subsection{Safe Traffic Flow}

Consider the scenario of a platoon of vehicles in which the length of each vehicle is $l$ units, and the maximum deceleration possible by a vehicle is $a_{min}.$  As illustrated in Figure~\ref{fig:dis}, a safe operating point of this platoon would be to ensure that for any vehicle $i,$ if the vehicle ahead of it, $i-1$ were to stop instantaneously, then vehicle $i$ must be able to come to a stop without hitting vehicle $i-1.$  This is a conservative approach, but provides a model under which collision-free operation can be deterministically guaranteed as long as the safety criterion is not violated.
\begin{figure}
\begin{center}
\includegraphics[width=0.4\textwidth]{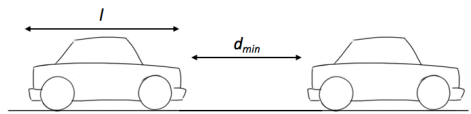}
\caption{The Eulerian view of two successive vehicles.}
\label{fig:dis}
\vspace{-.2in}
\end{center}
\end{figure}

Since the events of interest to us such as traffic merger requires that platoons be shaped at a particular location in space, we use the so-called \emph{Eulerian viewpoint} of the platoon in which we specify the state by observing the system from a fixed location in space $s,$ (as opposed to the \emph{Lagrangian} viewpoint which does so by moving with the vehicle of interest and specifying state as a function of time $t$.)   Thus, we use a fixed Cartesian coordinate system to indicate the position of each vehicle at each time.  Assume that all vehicles are traveling at velocity $v,$ and that vehicle $i-1$ has just passed location $s.$  Consider the scenario in which the distance between the vehicles $i-1$ and $i$ is $d.$  The \emph{time-gap} between the two vehicles, denoted as $\tau$ is the time required for vehicle $i$ to just pass location $s.$  The distance traveled by the vehicle $i$ during this time interval is $d + l.$ 

Determining the safe operating points in the system requires a kinematic model of the vehicles. For simplicity of exposition we use a so-called second order model of vehicular dynamics under which each vehicle can directly control its acceleration (within limits). Such a model does not account for the lag between the application of a control input and the desired acceleration actually being realized. However, it is relatively straightforward to extend the results that we present in this paper to a third order model that incorporates lag. Furthermore, since lag is a random variable that depends on individual vehicle dynamics and roadway conditions, the best option is often to develop controllers for the second order model, and then study their performance using a third order model (with random lags) numerically.

Let $d_{min}$ be the smallest distance between the two vehicles such that we satisfy our safety requirement at velocity $v.$  Then denoting the absolute value of the maximum possible deceleration of the vehicle by $a_{min}$, the minimum distance for safe operation is
\begin{equation}
d_{min} = \frac{v^2}{2a_{min}}.
\label{eq:dmin}
\end{equation}

Assuming vehicle $i$ has constant velocity $v$, the kinematic relation between $v,$ $\tau$ and $d_{min}$ is
$$v \tau = d_{min} + l,$$
where $\tau$ is minimum time-gap required for safety.  Then from (\ref{eq:dmin}) we have
the relation between time-gap and velocity for safe operation as
\begin{equation}
\tau = \frac{v}{2a_{min}} + \frac{l}{v}.
\label{eq:tv}
\end{equation}

\subsection{Need for Traffic Shaping}
Figure~\ref{fig:tv} provides an illustration of the relationship derived above. The length of each vehicle, $l$ is chosen as $6$ m (and conservatively represents both the actual vehicle length plus the standstill spacing desired) and $a_{min}$ is chosen to be $4$ m/s$^2$ in this example.  The x-axis represents the velocity $v$. The y-axis represents minimum time-gap $\tau$ achievable for each velocity. There is no control law that can safely operate (guarantee no collisions) outside of this convex region, i.e. the safe region is above the curve.  However, there could be a control law that guarantees safe operation for points inside the safety region.
\begin{figure}
\begin{center}
\includegraphics[width=0.4\textwidth]{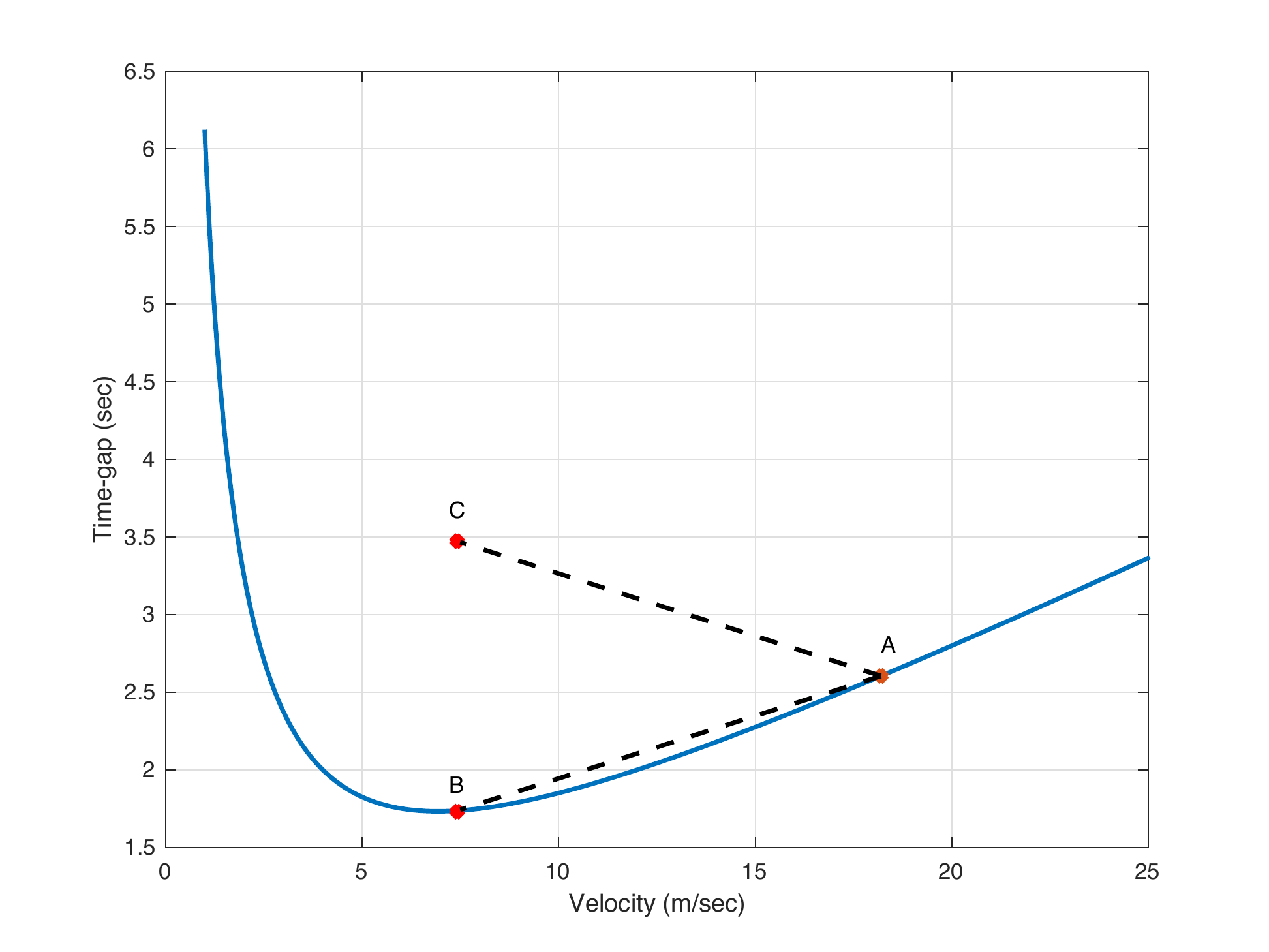}
\caption{Safety Region in terms of Time-Gap and Velocity.}
\label{fig:tv}
\end{center}
\vspace{-.1in}
\end{figure}

\begin{figure}
\begin{center}
\includegraphics[width=0.45\textwidth]{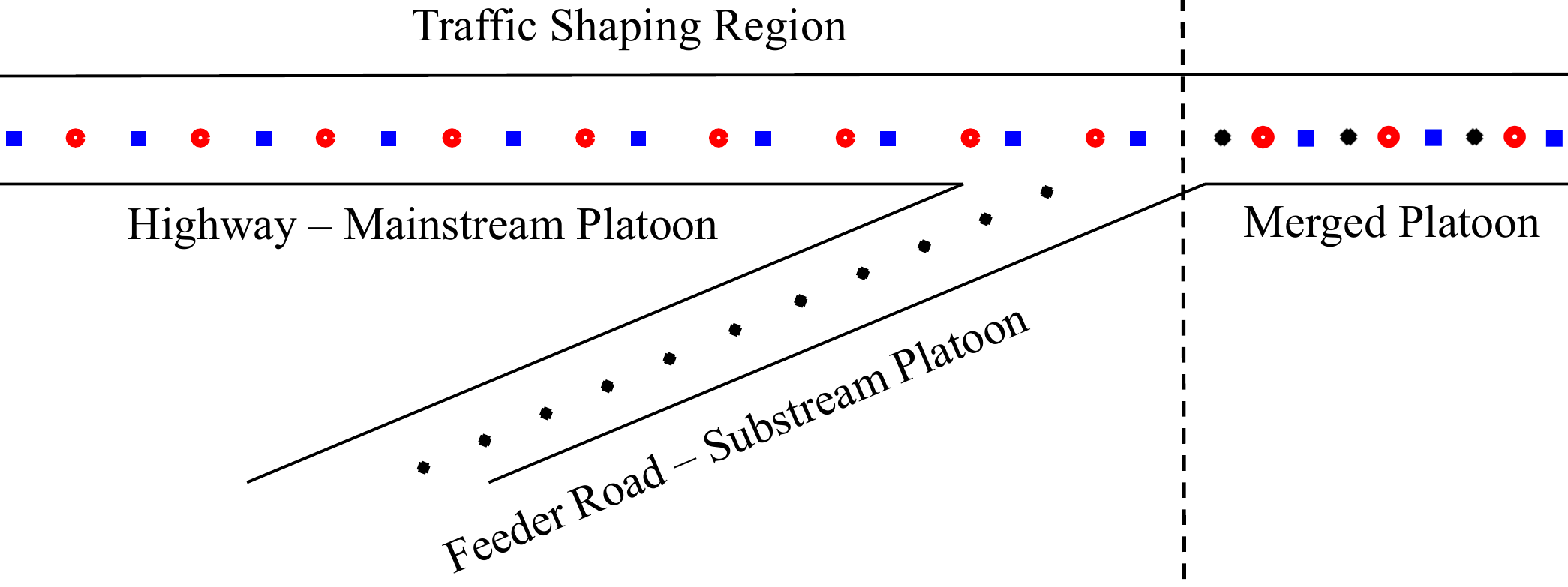}
\caption{Traffic shaping for smooth platoon merger.}
\label{fig:mot2}
\vspace{-.1in}
\end{center}
\end{figure}

Now, suppose that we have two platoons with different flow rates as in Figure \ref{fig:mot2}.  We refer to the existing platoon on the highway as the mainstream platoon, and the entering platoon from the feeder road as the substream platoon.    We wish to merge these two asymmetric platoons, in which one vehicle from the substream enters for  each pair of vehicles in the mainstream.  Assume that the initial (upstream) operating point of the mainstream is at $A = (v_1,\tau_1).$   Now, the minimum time-gap that can be supported in the merged platoon corresponds to the point $B = (v_2,\tau_2)$.   Suppose that the entering platoon from the feeder is already at $v_2.$    
In order to merge the two platoons, we need to first shape the mainstream into sub-platoons such that the velocity of all the vehicles is reduced to $v_2,$ and the time-gaps between vehicles in each sub-platoon are selected appropriately to ensure that the merged platoon operates at point $B = (v_2,\tau_2)$.

Figure~\ref{fig:intro} provides an illustration of the above idea in which all vehicles of the mainstream platoon initially operate at point $A = (v_1,\tau_1).$   The graph shows the desired velocity of every vehicle as a function of its location, and below that we provide a snapshot of how the vehicles are positioned in space at a particular time instant.  As the vehicles of the mainstream platoon proceed along the highway (left to right in the figure), they are shaped into sub-platoons, each of size two.  The lead vehicle of each sub-platoon (whose indices are even numbers) eventually operates at point  $C = (v_2,\tau_3),$ while the following vehicle in each sub-platoon (whose indices are odd numbers) operates at $B = (v_2,\tau_2).$  Notice that by shaping the mainstream platoon in this manner, we have created the necessary spacings such that the substream  platoon can be merged into the gaps created, with all vehicles in the merged platoon operating at point $B = (v_2,\tau_2).$   The procedure for merging shaped platoons is relatively straightforward, and is not the focus of the paper.  Our main goal is the design of a controller that can undertake traffic shaping of a single platoon.
\begin{figure}
\includegraphics[width=0.5\textwidth]{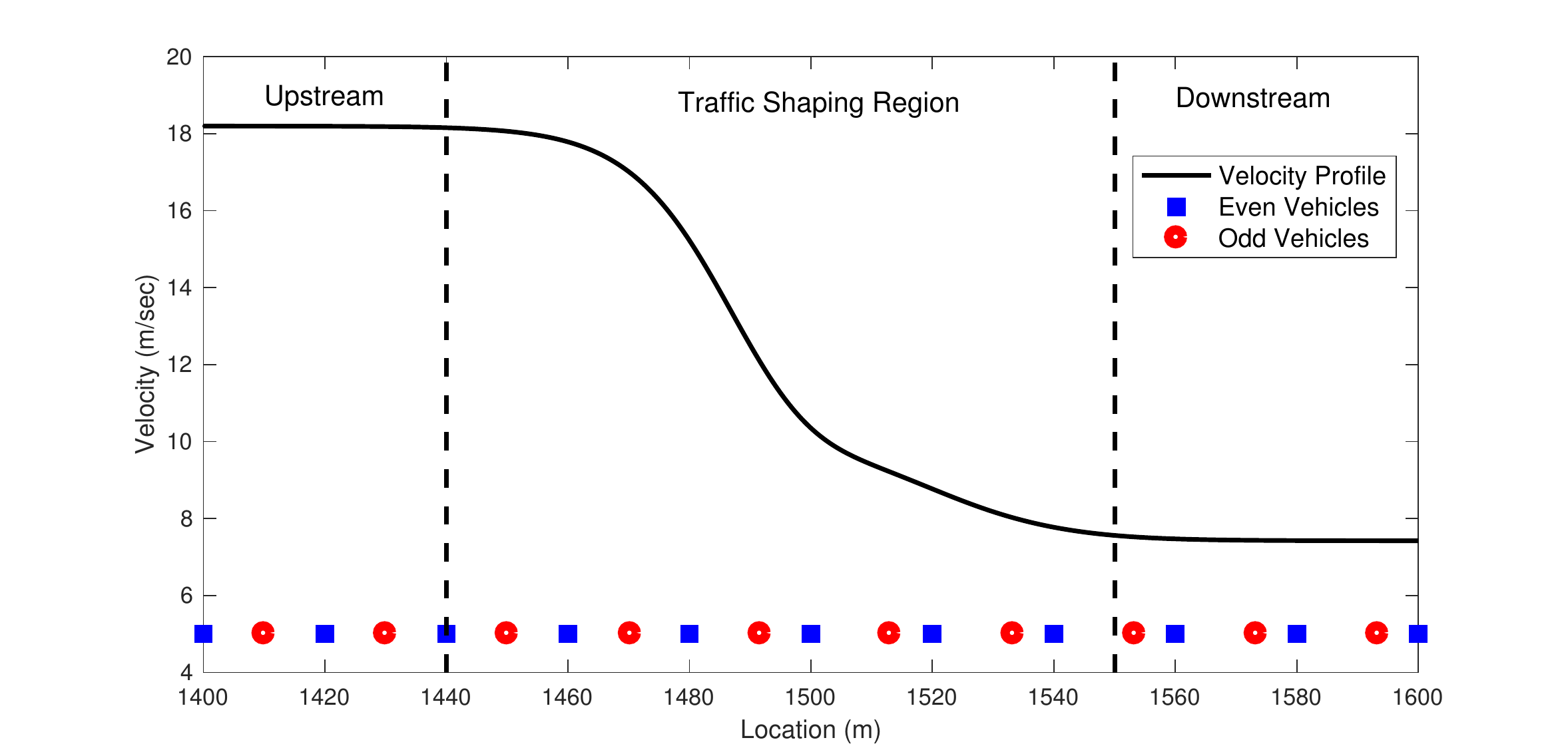}
\caption{Traffic shaping a platoon into sub-platoons of size two.}
\label{fig:intro}
\vspace{-.1in}
\end{figure}

The shaping of a given platoon to attain a desired operating point requires the specification of the time-gap and velocity as a function of location, which we refer to as a to as a \emph{time-gap and velocity profile}.  As in the example above, these desired profiles could potentially be different for each vehicle, and need to be selected carefully so that they can be realized using a simple controller.  In our example,  even vehicles move from operating point A to C , while odd vehicles move from A to B.  Hence, we require different desired time-gap profiles $\tau_{i,des}(s),$ based on whether the vehicle index is even or odd.    Since the target velocity of vehicles is identical, the desired velocity profile, $v_{des}(s)$ is the same for all of them. 
We will discuss how to design $\tau_{i,des}(s)$ and $v_{des}(s)$ in detail in Section~\ref{sec:prof}.

\subsection{Related Work}
There is much recent work on autonomous vehicle control. For instance, \cite{s1998} designs spacing policies that could lead to higher roadway throughput or reduce the congestion.  Most of these studies consider the Lagrangian viewpoint (i.e., from the perspective of a moving vehicle), and focus on the stability of traffic policies under a fixed spacing between vehicles (constant distance headway) or fixed time to collide with the preceding vehicle (constant time headway). For instance, the stability of a platoon under decreased time headway is studied in \cite{time1,time2,time3}.  While \cite{time1,time2} focus on performance of constant time-headway policy, \cite{time3} states that variable time-headway policy outperforms previous one in terms of traffic flow.

Distance-based spacing policies are studied in  \cite{comp,space,space2,distance1}.   In \cite{comp} it is indicated that a distance-based spacing policy leads to higher throughput than time-headway policy, however, at the cost of more inter-vehicular communication.  Constant distance headway is considered in \cite{space,space2} , where the focus is on stability of either homogeneous or heterogeneous platoons of vehicles.  In \cite{distance1}, it is proven that variable distance headway as a function of velocity leads to safer and higher throughput than constant time-headway.

The above works are interested in traffic flows as a function of time, and do not study traffic shaping as a function of location.  However, as we indicated in our motivating example, events such as lane drops and mergers require a platoon to be at a specific operating points at a particular location.  The first work that we are aware of that considers the Eulerian viewpoint is \cite{Bessel}, which introduces a constant time-gap policy.  They show that a platoon following a constant time-gap policy can be string stable. 

In the context of platoon merger, \cite{merge1,merge2} propose communication-based solutions to create gaps in a platoon into which another platoon maybe merged.  However, stability is not guaranteed through their communication process.  

The goal of the above work is to maintain constant headways in a platoon, which implies a fixed flow rate, and only allows change of the operating point along the velocity axis in Figure~\ref{fig:tv}.   Thus, none of the above-mentioned articles provide results on shaping traffic along the roadway to create variable gaps as a function of location that we desire.   
 

\subsection{Main Results}

Our paper has two main contributions.  The first is the design of a string-stable controller that is capable of tracking a multi-dimensional target profile (both time-gap and velocity) at each vehicle.  The technique that we use to obtain the controller is based on feedback linearization \cite{khalil}, and the controller itself can be considered as a generalization of the one presented in \cite{Bessel} in a high level sense, wherein space (not time) is used as the independent variable.  As mentioned above, the basis for such design is the Eulerian view of the system that permits us to connect the safety of flows (which must hold at all locations) with the idea of changing flows via traffic shaping (also as a function of location).   We show that each vehicle is able to track the desired profile, and that the platoon as a whole is string stable.

Our second contribution is in determining the nature of the feasible traffic profile under safety constraints.  
We use the context of traffic shaping for platoon merging, where greater flow has to be supported after the merger.  Further, the actual traffic shaping must happen is as short a distance as possible, since slowing down the vehicles far ahead of the merger point would cause increased travel times.  Hence, we pose the problem of profile design as an optimization problem in which the traffic shaping region must be minimized subject to safety constraints, as well as the deceleration limits of the vehicles.  We show how these boundary conditions imposed by safety considerations can be used in the design of time-gap and velocity profiles that can be attained by our controller.  Finally, we use numerical simulations using our combination of profile and controller to illustrate how it is successful in attaining the desired traffic shaping.

\section{System Model and Control Law}
\label{sec:control}
We now design a control law to attain a given time-gap and velocity profile. This section deals with the system model and controller design that is capable of tracking the desired target, while the next section will introduce a  (modified) concept of string stability and conduct stability analysis.

Consider an infinite platoon of homogenous vehicles. As described in the previous section, we consider a second-order kinematic model of vehicles as opposed to a third-order model that includes lag such as that considered in \cite{model1, model2}. However, essentially the same approach would work for the for third-order model as well, although the resultant controller would be more complex and less intuitive.   We begin with the basic second order kinematic equations
\begin{subequations}
\begin{equation}
\frac{ds_i}{dt}=v_i(t),
\end{equation}
\begin{equation}
\frac{dv_i}{dt}=u_i(t),
\end{equation}
\end{subequations}
where $s_i(t)$ is location of $i^{th}$ vehicle at time $t$, $v_i(t)$ is velocity of the vehicle at time $t$ and $u_i(t)$ is input to the vehicle $i$. The initial location of $i$ starts from $0$.  Index $i=0$ indicates the leading vehicle.  Notice that time in the independent variable in the above model.   

Supposing that velocities are positive, it is possible to change the independent variable to location $s$ as follows:
\begin{subequations}
\begin{equation}
\frac{dt_i}{ds}=\frac{1}{v_i(s)},
\end{equation}
\begin{equation}
\frac{dv_i}{ds}=\frac{u_i(s)}{v_i(s)},
\end{equation}
\label{eq:stat1}
\end{subequations}
where $t_i(s)$ is the time instant when vehicle $i$ passes location $s.$  Therefore, the time-gap between two vehicles is $t_i(s)-t_{i-1}(s).$ 

For the infinite platoon described above, the lead vehicle ($i=0$) needs to only track a given velocity profile.  However, other vehicles have to track the velocity profile as well as their time-gap profile.  We can define the state of each vehicle in terms of the errors that it sees between the desired targets and the actual value. Then, considering (\ref{eq:stat1}), we define two tracking error terms for each vehicle as follows:
\begin{subequations}
\begin{equation}
e_i(s)=\frac{1}{v_i(s)}-\frac{1}{v_{des}(s)},
\label{eq:vel-er}
\end{equation}
\begin{equation}
\Delta_i(s)=t_i(s)-t_{i-1}(s)-\tau_{i,des}(s),
\label{eq:delta1}
\end{equation}
\label{eq:stats}
\end{subequations}
where $v_{des}(s)$ is the desired velocity profile and $\tau_{i,des}(s)$ is the desired time-gap profile. 

We make following assumption on the time-gap profile.
\begin{assumption}
$\tau_{i,des}$ is twice continuously differentiable with respect to $s$. 
\end{assumption}

Effectively, each vehicle tracks its own time-gap profile.  An appropriate pattern of time-gap tracking divides the platoon to sub-platoons. Therefore, there are sub-platoons with smaller time-gap followed by another identical sub-platoon. Choosing the proper form will ensure string stability as will be discussed in the next section.

Now, we will derive control laws for the lead vehicle and the following vehicles. Since there is no time-gap error for the lead vehicle, the dynamics for the lead vehicle is
\begin{equation}
\frac{de_0}{ds}=-\frac{u_0(s)}{v^3_0(s)}-\frac{\partial (1/v_{des})}{\partial s}.
\label{eq:seq}
\end{equation}

Applying feedback linearization \cite{khalil} (i.e., choosing a control input such that the system is linearized) on the dynamics of the lead vehicle using control input 
\begin{equation}
u_0(s) = v^3_0\left[ pe_0(s)-\frac{\partial (1/v_{des})}{\partial s}\right],
\label{eq:u0}
\end{equation}
where $p > 0$ leads to
\begin{equation}
\frac{de_0(s)}{ds}=-pe_0(s).
\label{eq:e0-con}
\end{equation}
Hence, the dynamics of the lead vehicle is linear.

For other vehicles, we need to differentiate the error term (\ref{eq:delta1}) twice to obtain the input to the vehicles. Applying feedback linearization, the control input to such a following vehicle would be
\begin{equation}
u_i(s)=v^3_i(s)\left[ p_0\Delta_i+p_a\frac{d\Delta_i}{ds}+\frac{u_{i-1}(s)}{v^3_{i-1}(s)}-\frac{\partial^2\tau_{i,des}}{\partial s^2}\right],
\label{eq:ui}
\end{equation}
where $p_0 > 0$ and $p_1 > 0$. With V2V/V2I communication it is possible to obtain information about the preceding vehicle. Thus, the the dynamics of a following vehicles is
\begin{equation}
\frac{d^2\Delta_i}{ds^2}=-p_0\Delta_i-p_1\frac{d\Delta_i}{ds},
\end{equation}
where $p_0 > 0$ and $p_1 > 0$. 

Choosing $e_i(s)$ as the output of the each vehicle, $e_i(s)$ for following vehicles would be
\begin{equation}
e_i(s)=\frac{d\Delta_i}{ds}+e_{i-1}(s)+\frac{\partial\tau_{i,des}}{\partial s},
\end{equation}
since $\frac{d\Delta_i}{ds}=\frac{1}{v_i}-\frac{1}{v_{i-1}}-\frac{\partial\tau_{i,des}}{\partial s}$.

Using the inputs in (\ref{eq:u0}) and (\ref{eq:ui}), the dynamics of the platoon would be
\begin{subequations}
\begin{equation}
\frac{de_0}{ds}=-pe_0(s),
\label{eq:lead}
\end{equation}
\begin{equation}
\frac{d^2\Delta_i}{ds}=-p_0\Delta_i(s)-p_1\frac{d\Delta_i}{ds},
\label{eq:delta}
\end{equation}
\begin{equation}
e_i(s)=\frac{d\Delta_i}{ds}+e_{i-1}(s)+\frac{\partial\tau_{i,des}}{\partial s}.
\label{eq:ei}
\end{equation}
\label{eq:dynam}
\end{subequations}
where $p > 0$, $p_0 > 0$ and $p_1 > 0$.

\section{Stability Analysis}
\label{sec:stability}

In this section, we will prove the stability of our controllers.  We have two notions of stability that will be required at an individual vehicle level (plant stability), and at the level of a string of
vehicles (string stability).  


A second assumption on the time-gap profile will be required to
show stability.
\begin{assumption}
$\lim_{s\to\infty} \frac{\partial\tau_{i,des}}{\partial s}=0$.
\end{assumption}

\subsection{Plant Stability}


\begin{theorem}
The unique equilibrium point of the dynamics defined in (\ref{eq:dynam}) is the origin ($0$ error).  The equilibrium point is asymptotically stable if the parameters $p$, $p_0$ and $p_1$ are chosen to be positive.
\end{theorem}
\begin{proof}
The proof is though induction over vehicles, since the error of each following vehicle $i$ is dependent on vehicle $i-1$ as given in (\ref{eq:dynam}).  First, we start with the lead vehicle dynamics  (\ref{eq:lead}), which is a first-order differential equation.  Solving, the velocity tracking error of lead vehicle is:
\begin{equation}
e_0(s)=e_0(0)\exp(-ps),
\label{eq:e0}
\end{equation}
for $s>0$ ,and where $e_0(0) = \frac{1}{v_0(0)} - \frac{1}{v_{des}(0)}$. Therefore,
$$\lim_{s\to\infty} e_0(s) = 0,$$
since $p>0$. Thus, the error term of lead vehicle is asymptotically stable.



Next, the solution to second-order equation (\ref{eq:delta}), which represents the dynamics of time-gap tracking error is
\begin{equation}
\Delta_i(s) = A_i\exp(r_1s)+B_i\exp(r_2s),
\label{eq:deltai}
\end{equation}
where $r_1$ and $r_2$ are roots of characteristic function of (\ref{eq:delta})
\begin{equation}
r_1 = \frac{-p_1-\sqrt{p_1^2-4p_0}}{2p_1},
r_2 = \frac{-p_1+\sqrt{p_1^2-4p_0}}{2p_1}.
\end{equation}
Also, $A$ and $B$ are determined by initial conditions via Cramer's rule as
\begin{equation}
A = \frac{\begin{vmatrix}
\Delta_i(0) & 1 \\
\frac{d\Delta_i(0)}{ds} & r_2
\end{vmatrix}}{\begin{vmatrix}
1 & 1 \\
r_1 & r_2
\end{vmatrix}} ,
B = \frac{\begin{vmatrix}
1 & \Delta_i(0) \\
r_2 & \frac{d\Delta_i(0)}{ds}
\end{vmatrix}}{\begin{vmatrix}
1 & 1 \\
r_1 & r_2
\end{vmatrix}},
\end{equation}
where $\Delta_i(0)$ is $t_i(0) - t_{i-1}(0) - \tau_{i,des}(0)$ and $\frac{d\Delta_i(0)}{ds}=\frac{1}{v_i(0)} - \frac{1}{v_{i-1}(0)} - \frac{\partial \tau_{i,des}(0)}{\partial s}$.  Hence, according to (\ref{eq:deltai}):
$$\lim_{s\to\infty}\Delta_i(s)=0,$$
since $\Re(r_1) , \Re(r_2) < 0$. Thus, $\Delta_i$ of each vehicle is asymptotically stable.

Finally, the error in tracking velocity of vehicle $i$ is found by solving (\ref{eq:ei}).  Applying induction on (\ref{eq:ei}), we obtain
\begin{equation}
\begin{multlined}
e_i(s) = e_0(s)+r_1\exp(r_1s)\sum_{j=1}^i A_j +\\
 r_2\exp(r_2s) \sum_{j=1}^i B_j + \sum_{j=1}^i \frac{\partial \tau_{j,des}}{\partial s}.
\end{multlined}
\end{equation}
Now, $e_0(s)$ is asymptotically stable, as proven in (\ref{eq:e0}).  The second and third terms also go to $0$ as $s\to\infty$, since $\Re(r_1), \Re(r_2)<0$.  Further, from Assumption~2,   $\tau_{i,des}(s)$, $\lim_{s\to\infty}\frac{\partial \tau_{i,des}}{\partial s} = 0$. Hence, 
$$\lim_{s\to\infty}e_i(s) = 0,$$
which shows that $e_i(s)$ is asymptotically stable.
\end{proof}

\subsection{String Stability}

In the previous subsection, we proved plant stability, meaning that each vehicle is able to track desired velocity and time-gap profiles. However, we need to guarantee that an error that occurs in one of the vehicles does not amplify unboundedly through the entire platoon. 

This system is supposed to track both velocity and time-gap profiles. Therefore, string stability must be studied for both velocity tracking error and time-gap tracking error.  According to (\ref{eq:delta}), the error of time-gap tracking is decoupled, i.e. there is no effect of proceeding vehicle on the vehicle.  Therefore, any error that occurs at any vehicle regarding time-gap tracking, would not propagate through the whole platoon.  Further, since each vehicle is asymptotically stable in terms of $\Delta_i$, this error term would damp in each vehicle.  Hence, in terms of time-gap tracking error, the dynamics of (\ref{eq:dynam}) is string stable in the notion of \cite{swaroop}.

However, as is obvious from (\ref{eq:ei}), the preceding vehicle affects the behavior of the following one. Thus, the system is not decoupled in terms of velocity tracking error.  The definition of string stability in this study is that errors should not be amplified unboundedly.  Hence, we require that given any $\epsilon >0$, there exists a $\delta >0$ such that
\begin{equation}
\Vert e_i(0)\Vert_{\infty} < \delta \Rightarrow \sup_i \Vert e_i(.)+\eta\Vert_{\infty} < \epsilon,
\label{eq:ourdef}
\end{equation}
where $\eta$ is constant.  In the above definition, $\kappa$ is an upper bound on acceptable velocity tracking error. The definition indicates that if an error occurs in one of the vehicles, will not amplify by more than $\eta$. 

\begin{theorem}The platoon of vehicles with dynamics (\ref{eq:dynam}) and 
\begin{equation}
\tau_{i,des}(s)=\tau_0+(-1)^iT_{des}(s) \label{eq:timegap-design}
\end{equation} is string stable with $$\eta=\vert\frac{\partial T_{des}}{\partial s}\vert.$$ 
$\tau_0$ is initial time-gap between vehicles($\tau_0=\lim_{s\to-\infty}\tau_{i,des}(s)$). $T_{des}(s)$ is a variable part of the time-gap profile and is assumed to be decreasing, twice continuously differentiable in $s$ and $\lim_{s\to\infty}\frac{\partial T_{des}}{\partial s}=0$.
\end{theorem}
\begin{proof} Consider the dynamics in (\ref{eq:dynam}) with $p>0$, $p_0>0$ and $p_1>0$, initial conditions of $\Delta_i(0)=\frac{d\Delta_i(0)}{ds}=0,$ and velocity tracking error $e_0(s)$ for the leading vehicle.  For all $s\geq 0,$ the errors $e_i(s)$ of following vehicles would satisfy
\begin{equation}
\vert e_i(s)\vert\leq \vert e_0(s)\vert +\vert \frac{\partial T_{des}}{\partial s}\vert.
\label{eq:str}
\end{equation}

In order to prove (\ref{eq:str}), consider the (\ref{eq:ei}) and
\begin{equation}
e_1(s)=e_0(s)+\frac{\partial T_{des}}{\partial s}.
\end{equation}

Thus, by induction we have
\begin{equation}
e_i(s)=e_0(s)+\frac{\partial T_{des}}{\partial s} \sum_{j=1}^{i} (-1)^j.
\label{eq:1}
\end{equation}

Since sequence $(-1)^i$ is periodic, $\sum_{j=1}^{i} (-1)^j$ is bounded by $1$. Further, applying the triangle inequality on (\ref{eq:1}) yields the inequality (\ref{eq:str}).  Hence, it is obvious that $\eta$ in (\ref{eq:ourdef}) is $\vert e_0(s)\vert + \vert\frac{\partial T_{des}}{\partial s}\vert$, which means that $\eta=\vert\frac{\partial T_{des}}{\partial s}\vert$.  Consequently, the dynamics in (\ref{eq:dynam}) is string stable in terms of both errors of tracking the velocity and time-gap profiles.
\end{proof}

\section{Profile Design}
\label{sec:prof}

In this section our objective is to design time-gap and velocity profiles that will be used as inputs to the control laws derived in Section~\ref{sec:control}.  There are three conditions that must be satisfied during the design process.  First, the traffic shaping procedure must be undertaken in the smallest length of roadway, since this would allow for vehicles to travel at the (high) initial velocity for the longest distance, and hence minimize the travel time.   Second, the deceleration required of the vehicles must remain bounded by the maximum acceptable value.  Finally, the time-gaps and velocities of each vehicle at any location must lie in safety region shown in Figure \ref{fig:tv}. 

\begin{figure*}[htbp]
\centering
\begin{minipage}{.3\textwidth}
\centering
\includegraphics[width=1\columnwidth]{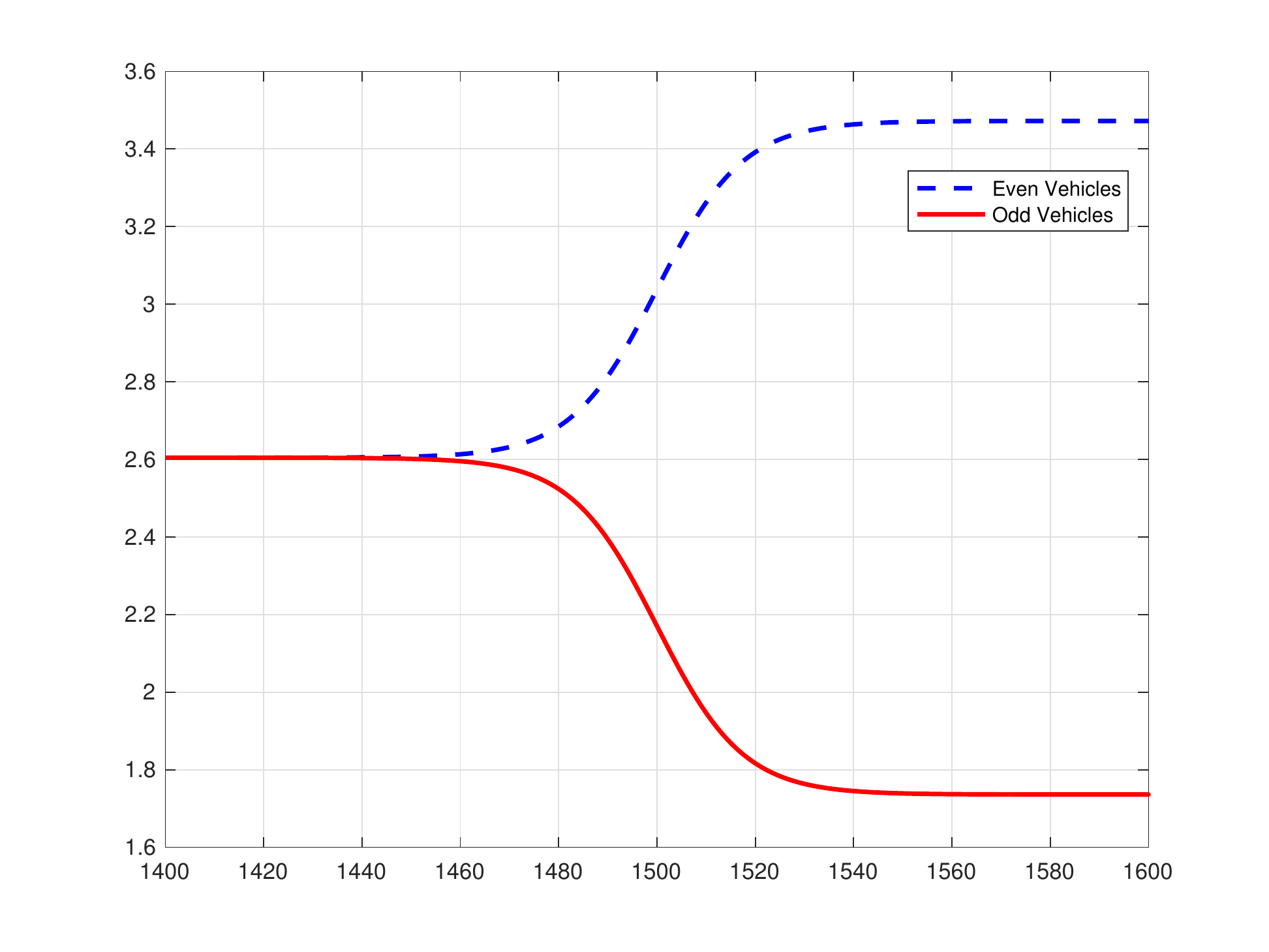}
\caption{Time-Gap Profile v.s. Location}
\label{fig:time-prof}
\end{minipage}\hfill
\begin{minipage}{.3\textwidth}
\centering
\includegraphics[width=1\columnwidth]{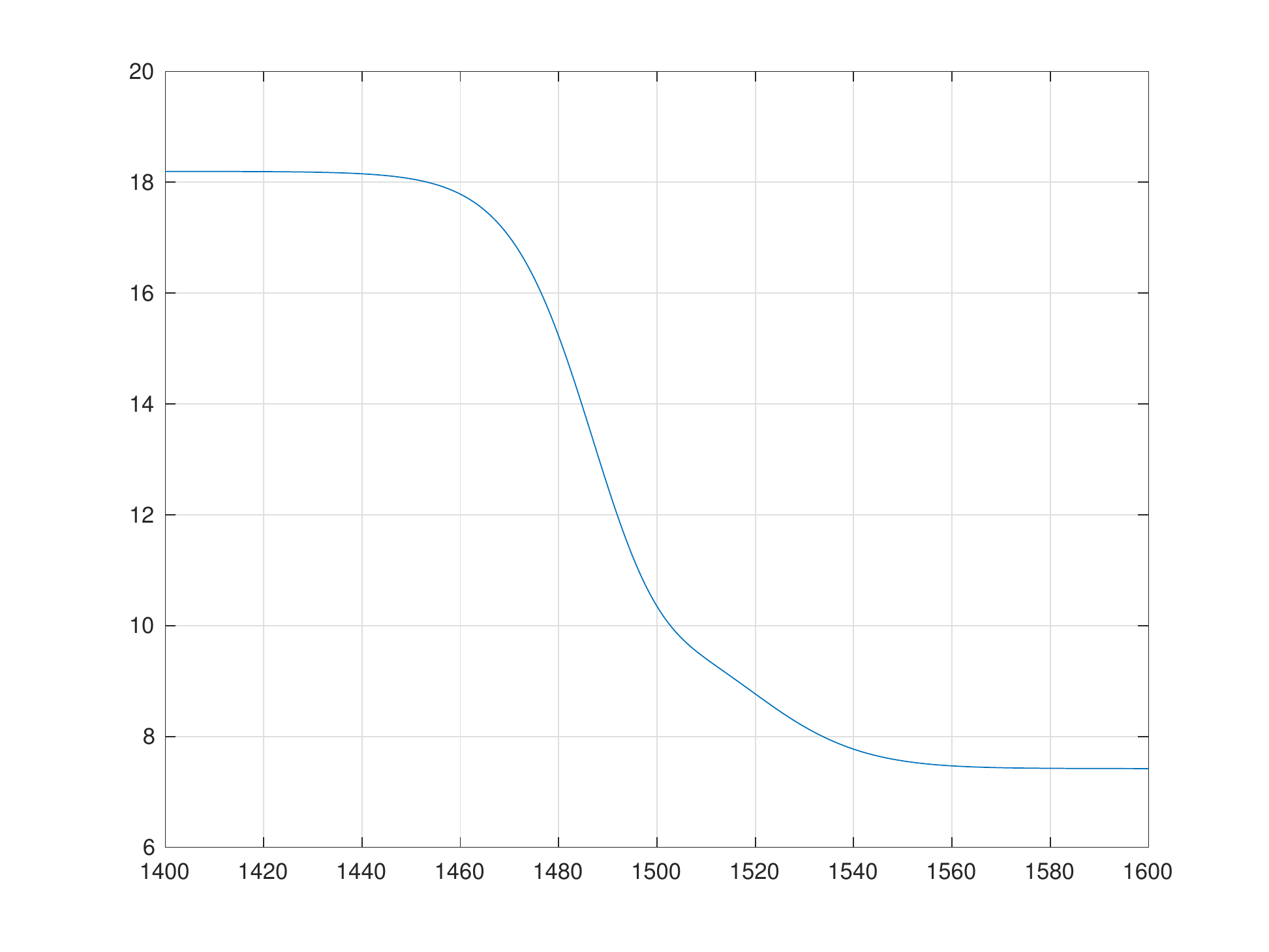}
\caption{Velocity Profile v.s. Location}
\label{fig:vel-prof}
\end{minipage}\hfill
\begin{minipage}{.3\textwidth}
\centering
\includegraphics[width=\columnwidth]{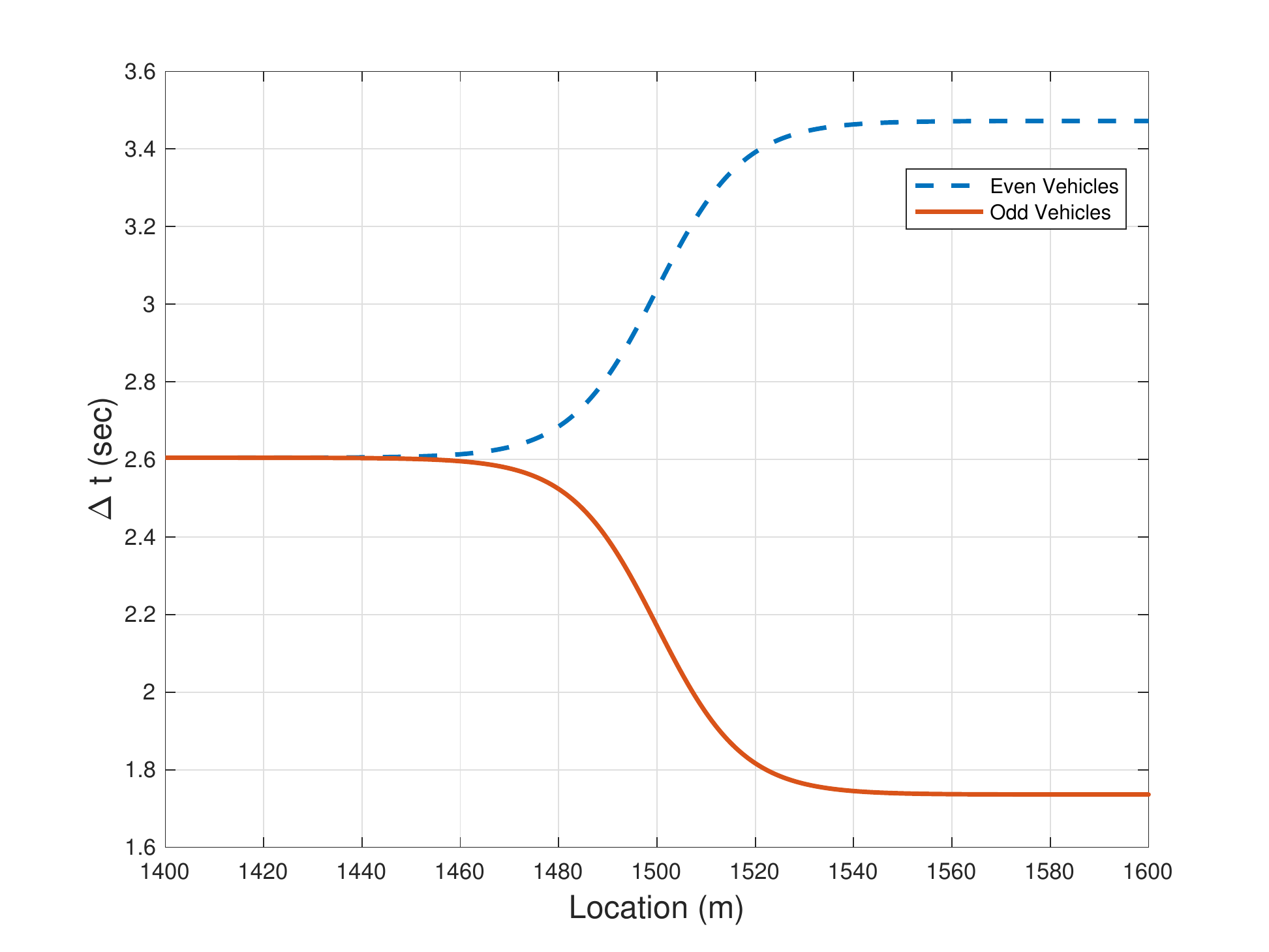}
\caption{Realized Time-Gap Profiles.}
\label{fig:tg}
\end{minipage}
\end{figure*}
\subsection*{Time-Gap Profiles}

Our first problem is to design a time-gap profile that is simple (has few parameters over which to optimize), and satisfies Assumption 1 and Assumption 2. 
Our choice in this study is the logistic function. 
Hence, the desired time-gap profiles based on general form given in (\ref{eq:timegap-design}) are
\begin{equation}
\tau_{i,des}(s) = \tau_0 + (-1)^i (\alpha + \beta \tanh(\gamma s)),
\label{eq:tg-log}
\end{equation}
where $\gamma > 0$ and $\alpha$ and $\beta$ are to be chosen based on the desired initial and final value of time-gaps and the kinematic constraints.  The initial value of time-gap profile is $\lim_{s\to-\infty}\tau_{i,des}(s)$, which is $\tau_0$, and final value is $\lim_{s\to\infty}\tau_{i,des}(s)$.  For instance, if $\tau_{odd,end}$ is final value of time-gap of odd vehicles, $\alpha$ and $\beta$ would be:
\begin{equation}
\alpha =  \beta = \frac{\tau_0 - \tau_{odd,end}}{2}.
\label{eq:alpha}
\end{equation}

The parameter $\gamma$ determines the maximum slope of the logistic function.  In our platooning scenario, a larger slope intuitively implies a shorter distance required for traffic shaping.  However, we cannot choose $\gamma$ arbitrarily large, since that would impact the maximum deceleration required of the vehicles.  In order to understand this intuition precisely, we first need to determine the impact of the time-gap profile choice on the velocity profile.

\subsection*{Velocity Profile}

The velocity profile must be designed jointly with the time-gap profile so as to ensure that vehicles are within the safe region of operation (see Figure~\ref{fig:tv}) at all times.  Now, the odd vehicles are the ones that operate on the boundary of the safe region (the operating point for odd vehicles goes from point A to B in Figure~\ref{fig:tv}).  Hence, we design the velocity profile of odd vehicles in a manner that they operate on the boundary of the safe region as specified by (\ref{eq:tv}).  Thus, from (\ref{eq:tv}), the velocity profile of odd vehicles as a function of time-gap is
\begin{eqnarray}
v_{odd}(s) =  \hspace{2.5in} \nonumber \\ \left( \tau_{odd,des}(s) a_{min} +
 \sqrt{(\tau_{odd,des}(s) a_{min})^2 - 2la_{min}} \right),
\label{eq:v-t}
\end{eqnarray}
where $\tau_{odd,des}(s)$ is the time-gap profile for odd vehicles as determined from (\ref{eq:tg-log}).

We next find the velocity profile of even vehicles, which is the same as the lead vehicle desired velocity profile $v_{des}(s)$ (since the lead vehicle is even).  To determine the velocity profile, we consider  (\ref{eq:ei}) in which we set all the error terms equal to zero (since profile design is under the assumption of ideal operation).  Thus, we set $\frac{d\Delta_i}{ds}=e_{i-1}(s) =0.$  Hence, from (\ref{eq:vel-er}) 
\begin{equation}
 v_{des}(s) = v_{even}(s)=\frac{v_{odd}(s)}{1-v_{odd}(s)\frac{\partial T_{des}}{\partial s}}.
\label{eq:vdes}
\end{equation}


\subsection*{Optimization Problem}

Now that we have candidate time-gap profiles and the corresponding safe velocity profiles,  we have to determine the value of $\gamma$ so that the deceleration process occurs in as small as possible distance, while ensuring that deceleration remains bounded.  The following optimization problem may be used to determine the value of $\gamma$:
\begin{equation}
\begin{aligned}
& \text{maximize}
& & \gamma \\
& \text{s.t.}
& & \min a_{even} \geq a_{min} \\
&&& \min a_{odd} \geq a_{min} \\
&&& \gamma > 0.
\end{aligned}
\label{eq:opt}
\end{equation}


Since the constraints above are in the form of bounds on acceleration, we need to determine the acceleration requirements imposed by the candidate velocity profile (\ref{eq:vdes}). The relation between velocity and acceleration is easily determined via the chain rule as
$$\frac{\partial v}{\partial t} = \frac{\partial v}{\partial s} \frac{\partial s}{\partial t} = v \frac{\partial v}{\partial s}.$$
The equation above holds for both odd and even vehicles.  Thus, we can find $a_{odd}$ and $a_{even}$ by substituting for odd and even vehicles using (\ref{eq:v-t}) and (\ref{eq:vdes}),  respectively.  Now that we have a well defined set of constraints, we solve the optimization problem numerically. 

Once we determine $\gamma$, we immediately have the time-gap profile $\tau_{i,des}(s).$  We then use (\ref{eq:v-t}) to yield the odd vehicle velocity profile $v_{odd}(s)$. Finally, we find $v_{des}(s) = v_{even}(s)$ using (\ref{eq:vdes}). 

\begin{figure*}[ht]
\centering
\begin{minipage}{.3\textwidth}
\centering
\includegraphics[width=1\columnwidth]{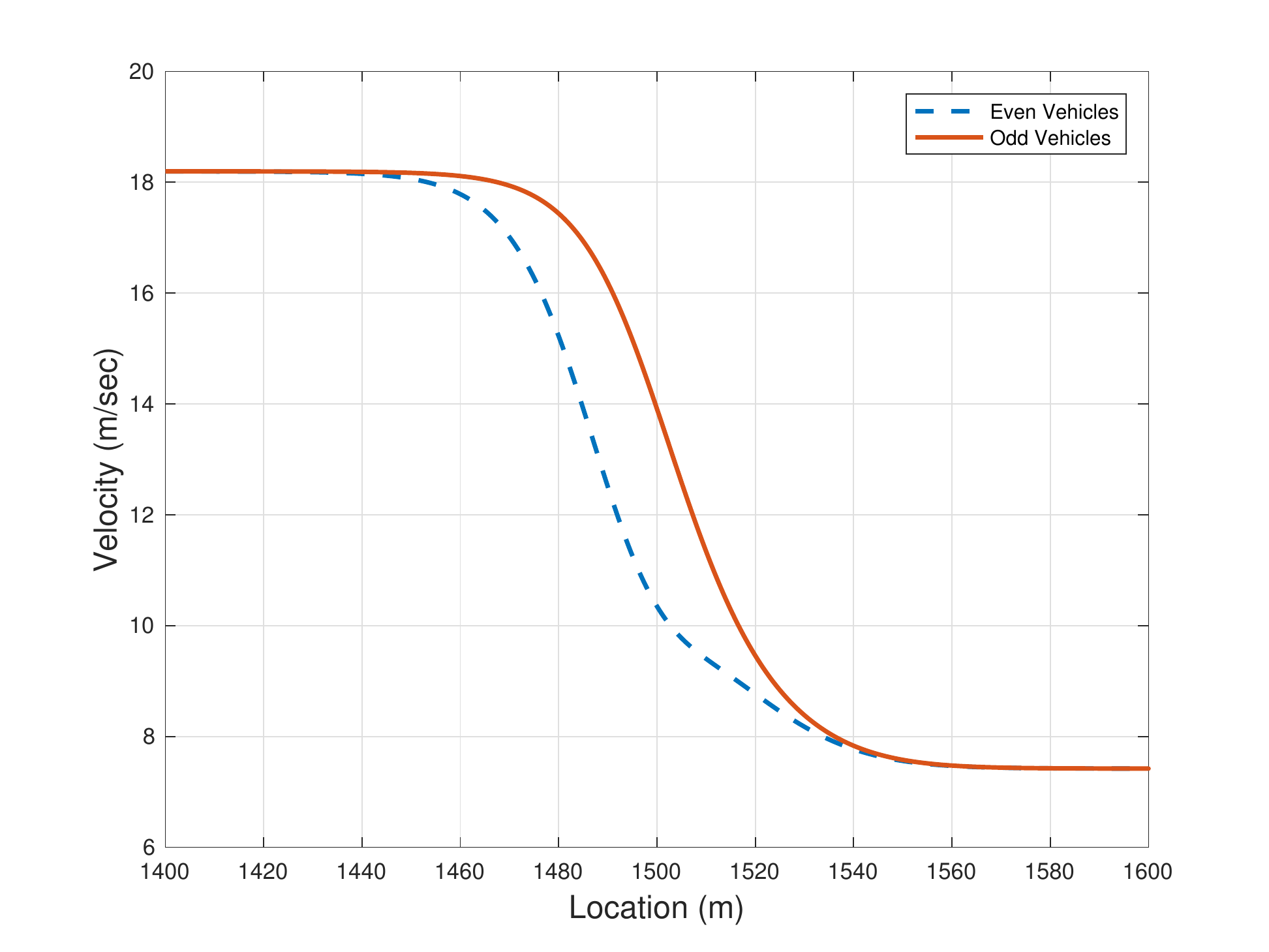}
\caption{Velocity of Two Successive Vehicles. The pattern is periodic.}
\label{fig:vel}
\end{minipage}\hfill
\begin{minipage}{.3\textwidth}
\centering
\includegraphics[width=1\columnwidth]{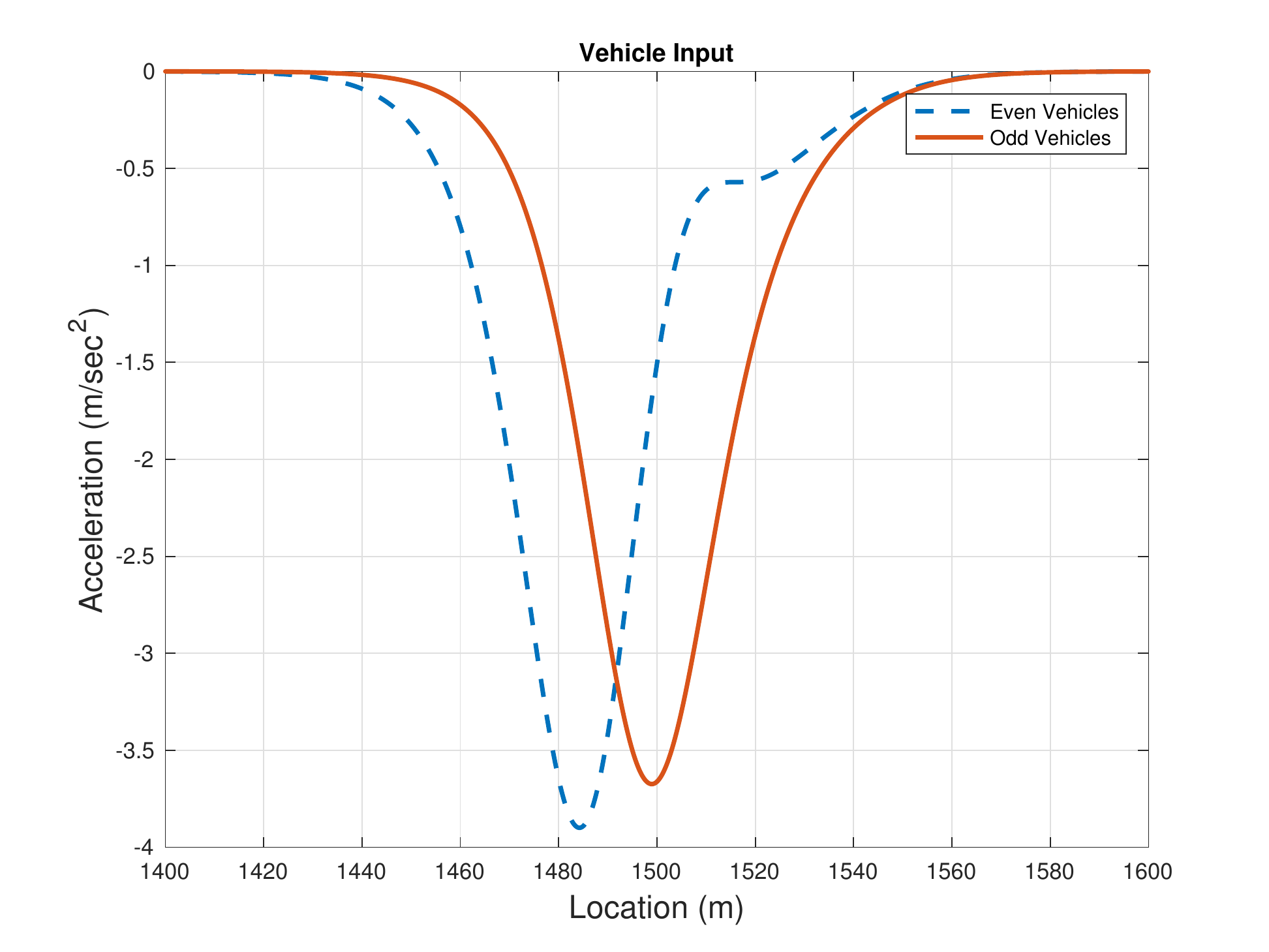}
\caption{Acceleration of Odd and Even Vehicles }
\label{fig:accel}
\end{minipage}\hfill
\begin{minipage}{.3\textwidth}
\centering
\includegraphics[width=1\columnwidth]{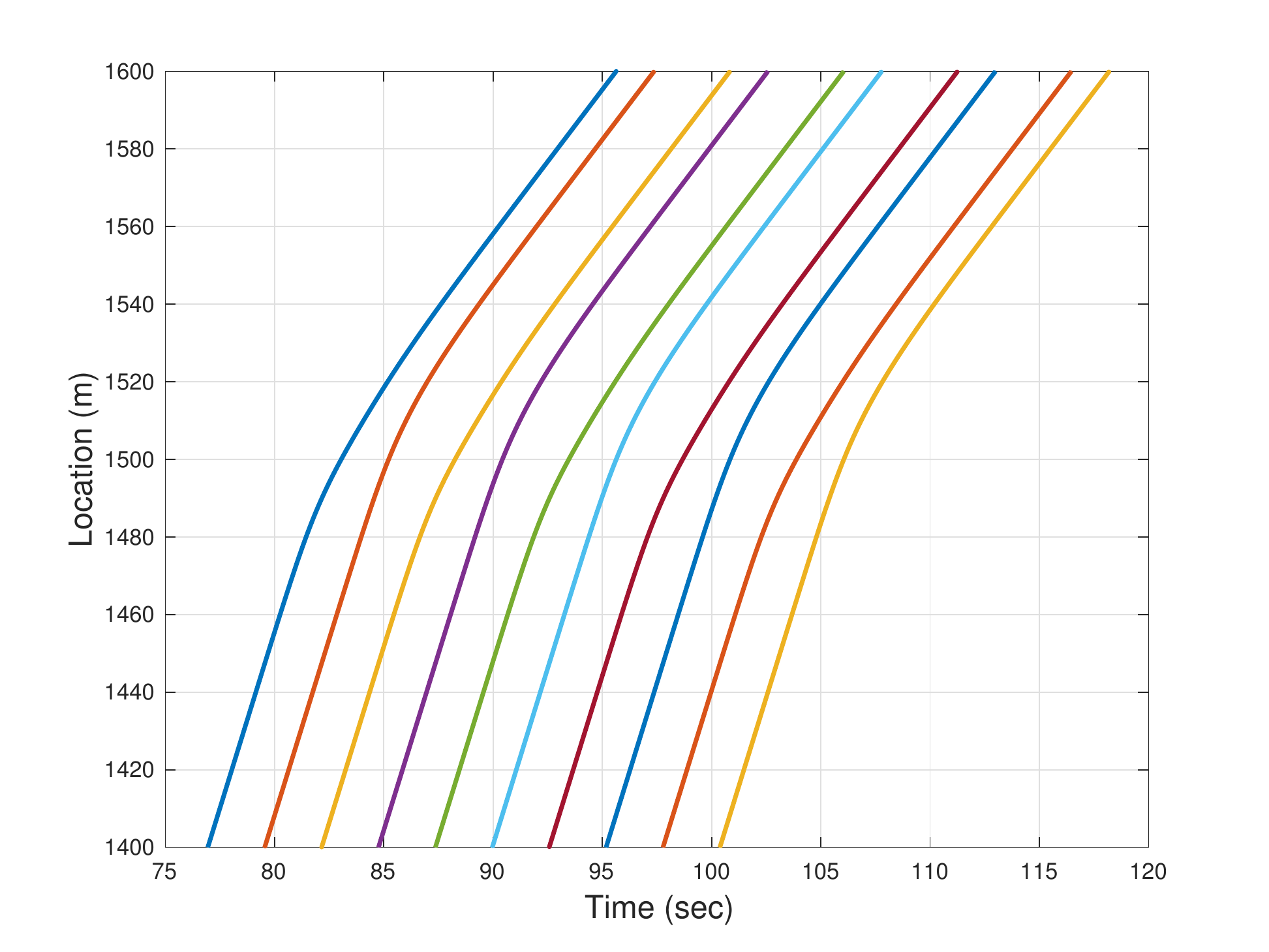}
\caption{Trajectories of all Vehicles in the Platoon}
\label{fig:tra}
\end{minipage}
\end{figure*}

\section{Evaluation}
\label{sec:eval}

We now perform numerical simulations to evaluate the performance of the proposed controller according to designed profiles.  We use the example of traffic shaping of a platoon of vehicles in the manner described in Section~\ref{sec:intro}.   We recall that the objective in that example was to reduce the operating velocity of the mainstream vehicles, while adjusting their time-gaps appropriately to allow for the merger with another substream platoon.  Thus, we desire to shape the mainstream platoon into sub-platoons with a smaller time-gap between members of each sub-platoon, and larger time-gaps between each sub-platoon.

We choose $\tau_0 = 2.6 sec$, which is $1.5$ time of minimum of time-gap in Figure \ref{fig:tv}.  $\tau_{odd,end} = 1.74 sec$ is the minimum time-gap supported by merged platoon on the highway.  Under these targets, $\alpha$ and $\beta$ in (\ref{eq:alpha}) would be $2.17$ and $0.43,$ respectively.   Finally, to determine parameter $\gamma$, we solve the optimization problem of (\ref{eq:opt}) numerically.  The optimal value of $\gamma$ is $0.057$.  Figure \ref{fig:time-prof} illustrates the time-gap of each vehicle at each location ($\tau_{i,des}(s)$), while Figure \ref{fig:vel-prof} shows the velocity profile ($v_{des}(s)$).

Figure~\ref{fig:tg} illustrates the time-gap tracking performance of the controller as a function of its location.   
There are two plots shown in this figure.  On the one hand, the decreasing plot is associated with odd vehicles, showing that they reduce the time-gaps and approach closer to the vehicle in front.  On the other hand, the increasing time-gap plot shows the time-gap of even vehicles, and indicates the creation of sub-platoons.  Comparison of the realized time-gaps of vehicles in Figure~\ref{fig:tg} with the desired time-gap profiles in Figure \ref{fig:time-prof} indicates precise tracking is achieved.

Figure~\ref{fig:vel} shows the realized velocity profiles of odd and even vehicles.vehicles along the road.  The figure indicates tha even vehicles are track the desired velocity profile shown in Figure~\ref{fig:vel-prof} accurately.  However, odd vehicles show some amount of error while decelerating.  This observation is consistent with (\ref{eq:ei}), which indicates that odd vehicles have an error of $\frac{\partial \tau}{\partial s}$ in tracking the desired velocity profile.  All vehicles eventually converge to an identical velocity downstream. 

The acceleration of all vehicles are in the desired bounded range at all locations as seen in Figure \ref{fig:accel}.  The acceleration of odd and even vehicles follow different profiles, consistent with their tracking different time-gap profiles as well as experiencing errors in velocity tracking in the case of odd vehicles.

Figure \ref{fig:tra} shows the trajectories of vehicles as a function of time.  As the figure indicates, the uniform spacing across vehicles in the platoon (bottom of figure) gradually gives way to sub-platoons of size two (top of figure).  We see that the density of vehicles at a sub-platoon level (as indicated by their spacing) increases, while the velocity of each vehicle decreases simultaneously.



Finally, Figure \ref{fig:safe} indicates that the traffic shaping operation occurs in the safe region operation, since the time-gap and velocity of all vehicles remain above the curve at all locations. 

\begin{figure}[h]
\begin{center}
\includegraphics[width=0.4\textwidth]{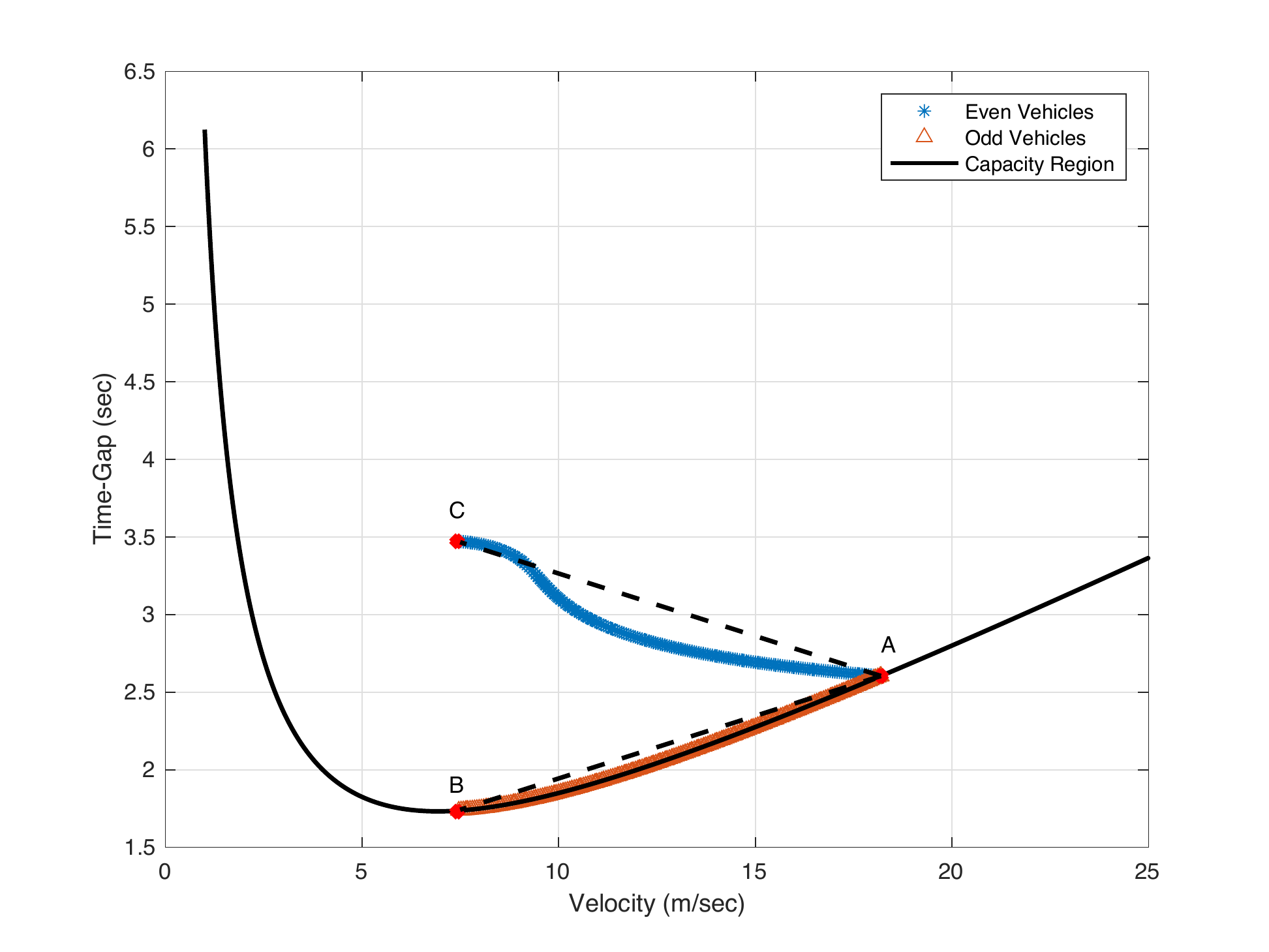}
\caption{Operation of controller in the safe region}
\label{fig:safe}
\end{center}
\end{figure}
\section{Discussion}
\label{sec:dis}
Merging two strings of vehicles is the motivation of this study as mentioned previously. There is an issue of providing gaps in between of vehicles and increasing the flow rate in order to prevent the shockwaves before merging. The proposed controller is capable of changing the distribution of vehicles in a string, which other controllers with constant time-gap or spacing policy are not. Thus, appending this controller with an appropriate merging policy could solve the automated merging problem.

\par An appropriate merging scenario with the controller could be locating two strings with a time-gap shift. Considering a road with two lanes merging, if the leading vehicles of two strings maintain a time-gap of $\tau_0$ and the strings change their distribution before merging, it is possible to locate the sub-platoons created to one platoon while merging. Therefore, after merging point there is a platoon of vehicles with higher flow rate and desired velocity. Besides, probable shockwaves are prevented, since all vehicles are tracking the velocity profile with respect to location. Consequently, this merging scenario shows the importance of the proposed controller.

\section{Conclusion}
\label{sec:conc}

In this work, we studied the problem of traffic shaping of a platoon of vehicles in terms of achieving a variable time-gap and velocity as a function of location.  The motivation for such shaping is to handle conditions of variable flow, such as two platoons merging due to a lane drop.   Existing controllers have focussed primarily on fixed spacing or time-gap regimes, which are not able to account for spatially changing flows.  The controller that we designed is able to locally modify a platoon by increasing or decreasing its flow locally via variable time-gaps in a provably safe manner.  In the example of platoon merger, it is thus able provide the necessary gaps to accomodate merging vehicles. 

Our approach involves defining system parameters over space, rather than time, and can be seen as an generalization using this approach for constant flow (fixed time-gap) problems.  Our methodological contribution was to identify the conditions needed on the time-gap and velocity profile to ensure stability, and to design such profiles keeping in mind the safe operating region of the system.   We showed that appropriately designed target profiles can be used to ensure that errors do not propagate unboundedly across vehicles, hence attaining a notion of string stability across the platoon.


\bibliographystyle{IEEEtran}
\bibliography{refs}

\end{document}